\documentclass[sigconf, authorversion]{acmart}

\usepackage{booktabs} 
\usepackage[ruled]{algorithm2e} 

\SetAlFnt{\small}
\SetAlCapFnt{\small}
\SetAlCapNameFnt{\small}
\SetAlCapHSkip{0pt}
\IncMargin{-\parindent}
\usepackage{subcaption}
\usepackage{optidef}
\usepackage{bbm}
\usepackage{tikz}
\usetikzlibrary{decorations.pathmorphing}
\usetikzlibrary{positioning}

\usepackage{natbib}
\usepackage{graphicx}
\usepackage{textcomp}
\usepackage{xcolor}
\usepackage{complexity}
\usepackage{url}

\copyrightyear{2024}
\acmYear{2024}
\setcopyright{rightsretained}
\acmConference[ICDCN '24]{25th International Conference on Distributed Computing and Networking}{January 4--7, 2024}{Chennai, India}
\acmBooktitle{25th International Conference on Distributed Computing and Networking (ICDCN '24), January 4--7, 2024, Chennai, India}\acmDOI{10.1145/3631461.3631546}
\acmISBN{979-8-4007-1673-7/24/01}

\begin{document}

\title{Parallel Algorithms for Equilevel Predicates}
\author{Vijay K. Garg} 
\affiliation{%
 \institution{The University of Texas at Austin}
 \city{Austin}
 \state{Texas}
 \postcode{78712}
 \country{USA}}
\email{garg@ece.utexas.edu}

\author{Robert P. Streit} 
\affiliation{%
 \institution{The University of Texas at Austin}
 \city{Austin}
 \state{Texas}
 \postcode{78712}
 \country{USA}}
\email{rpstreit@utexas.edu}

\newcommand{\exampleqedsymbol}{$\triangle$}

\newcommand{\remove}[1]{}
\newcommand{\h}{\hspace*{0.2in}}
\newcommand{\Ra}{\Rightarrow}
\newcommand{\ra}{\rightarrow}
\newcommand{\CC}{\mathcal{L}}
\newcommand{\CB}{{L_B}}
\newcommand{\CGG}{G}
\newcommand{\RR}{\mathbf{R}}

\bibliographystyle{unsrt}

\keywords{Detecting Predicates; Distributive Lattices; Equilevel Predicates}

\ccsdesc[500]{Theory of computation~Parallel algorithms}

\begin{abstract}
We define a new class of predicates called {\em equilevel predicates} on a distributive lattice which eases the analysis of parallel algorithms.
Many combinatorial problems such as the vertex cover problem, the bipartite matching problem, and the minimum spanning tree problem can be modeled as detecting an equilevel predicate.
The problem of detecting an equilevel problem is \NP-complete, but equilevel predicates with
the {\em helpful} property can be detected in polynomial time in an {\em online} manner. An equilevel predicate has
the {\em helpful} property with a polynomial time algorithm if the algorithm can return a nonempty set of indices such that advancing on any of them can be used to detect the predicate.
Furthermore, the refined {\em independently helpful} property allows online parallel detection of such predicates in \NC. When the {\em independently helpful} property holds, advancing on all the specified indices
in parallel can be used to detect the predicate in polylogarithmic time.

We also define a special class of equilevel predicates called {\em solitary} predicates.
Unless \NP = \RP, this class of predicate also does not admit efficient algorithms.
Earlier work has shown that solitary predicates with the {\em efficient advancement} can be detected
in polynomial time. We introduce two properties called the {\em antimonotone advancement} and
the {\em efficient rejection} which yield the detection of solitary predicates in \NC. 
Finally, we identify the minimum spanning tree, the shortest path,
and the conjunctive predicate detection as problems satisfying such properties, giving alternative certifications of their \NC~memberships as a result. 
 \end{abstract}

\maketitle

\section{Introduction}

We introduce a new class of predicates called {\em equilevel predicates} to model and solve many problems in parallel and distributed computing.
These predicates are defined on the elements of a distributive lattice \cite{davey}.
The distributive lattice is such that the height is small, but its size is large.
More concretely, this lattice can be viewed as generated from a 
poset with at most $n$ chains and $m$ elements per chain. 
In this case, the height of the lattice is bounded by $O(nm)$ but the cardinality can be as many as $n^m$ elements.
For efficiency, our methods operate on the underlying poset even though we are interested in finding an
element in the lattice satisfying the given predicate. This approach allows us to obtain algorithms with complexities efficient in $n$ and $m$.

\begin{figure*}
\begin{tabular}{ | l |  l  |  l | l |}
\hline
{\bf Class of Predicates} & {\bf Examples} & {\em Detection Algorithm}  \\
\hline
General Equilevel & Vertex Cover, $k$-conjunctive, Hamiltonian path & No efficient algorithm unless \P = \NP\\
\h With Helpful Property & bipartite matching, basis vectors & \P\\
\h With Independently Helpful Property & minimum spanning tree & \NC\\
\hline
General Solitary  & USAT & No efficient algorithm unless \NP = \RP \\
\h With Efficient Advancement & man-optimal stable marriage, housing allocation & \P\\
\h With Antimonotone Advancement & minimum spanning tree with unique weights & \NC~if the poset has $O(\log^k n)$ height\\
\h With Efficient Rejection & Graph Reachability, Conjunctive Predicates & \NC \\
\hline
\end{tabular}
\caption{Various Classes of  Predicates. Equilevel predicates are the ones that are true on elements of a lattice at a single level. Solitary predicates are the ones that are true on a
single element in the lattice.}
\end{figure*}

A predicate defined on a lattice is equilevel if all the elements of the lattice that satisfy the predicate
are on the same level of the lattice. It can be shown that detecting an equilevel predicate is, in general, a hard problem. For example, the minimum vertex cover problem for a graph can be modeled as detecting an equilevel predicate.
Even so, we will see that there are many equilevel predicates which can be detected efficiently in parallel, and that this efficiency can be implied by a few simple properties.

As a special case, we also consider the class of equilevel predicates which hold
on a single element in the underlying lattice. We call this class {\em solitary}. This class of predicates is related to the previously studied class of \emph{lattice-linear} predicates \cite{DBLP:conf/spaa/Garg20}, which are those predicates closed under the meet operation (or, the infimum operation) on the underlying lattice.
Such predicates have special properties that aid in the design of efficient detection algorithms.
Observe that for any 
lattice-linear predicate $B$ we can define a stronger predicate $B'$ that is satisfied only on the least element in the lattice satisfying $B$.
Then $B'$ is a solitary predicate. For example, consider the distributive lattice of assignments
for the stable marriage problem \cite{gale1962college,gusfield1989stable}. There may be multiple stable marriages; however, there is a unique
man-optimal stable marriage. The predicate that an assignment corresponds to the man-optimal
stable marriage is a solitary predicate.

In this paper, we investigate equilevel and solitary predicates, as well as the conditions giving efficient parallel algorithms for their detection. 
We first show that the problem of equilevel predicate is \NP-complete in general.
Moreover, the implications of this result are far-reaching in the landscape of lattice-linear predicate detection. For example, we show that slight generalizations of lattice-linear predicate-based problems with efficient solutions, such as
the conjunctive predicate detection,
are \NP-complete.
We then define a property on equilevel predicates called the {\em helpful}
property. 
An equilevel predicate has
the {\em helpful} property with a polynomial time algorithm if the algorithm can return a nonempty set of indices such that advancing on any of them can be used to detect the predicate in polynomial time in an online manner. 
We apply the helpful property to modeling and efficient detection of various problems such as bipartite matching and computing bases which span sets of vectors as equilevel predicates. Furthermore, when this is refined to the stronger
{\em independently helpful} property, there exists a simple \NC~algorithm
for problems modeled as equilevel predicates.

Next, a special class of equilevel predicates, called solitary predicates, is identified.
These are those predicates that can only be true on a single element in the lattice, and we show that even with this substantial restriction, detecting these 
predicates is still hard unless \NP=\RP. 
However, prior work implies that solitary predicates can be detected efficiently when they satisfy \emph{efficient advancement}\cite{DBLP:conf/spaa/Garg20}.
We refine the efficient advancement property in a few ways to obtain \NC~algorithms.
In particular, we show that whenever a problem satisfies the
{\em antimonotone} advancement property, it can be solved in \NC~time when the 
poset generating the lattice has a small height. For example, it can be shown that
the problem of the minimum spanning tree satisfies antimonotone advancement
property, which gives an alternative certification of its membership in \NC. Another special class of efficient advancement
property giving \NC~algorithms is the {\em efficient rejection} property. We show that
the shortest path problem \cite{Dijkstra1959}, and 
the conjunctive predicate detection in distributed computations \cite{GargWald:WeakUnstable} satisfy this property.

In summary, this paper makes the following contributions:
\vspace{-.025\baselineskip}
\begin{itemize}
\item 
The paper introduces the class of equilevel predicates.
Detecting a general equilevel predicate is \NP-complete. We show that a slight generalization
of the conjunctive predicate detection where one asks for a global state with exactly
$k$ events satisfying the conjunctive predicate is also \NP-complete.
\item
We show that any equilevel predicate with a helpful property (defined in this paper)
can be detected in polynomial time. The problem of finding the size of a maximum matching
in a bipartite graph falls in this class.
\item
We show that any equilevel predicate with the independently helpful property
can be detected in \NC. The problem of computing a minimum spanning tree in a weighted undirected
graph is shown to be in this class. 
\item 
The paper introduces a subclass of equilevel predicates called solitary
predicates. We show there is no randomized polynomial time
detection algorithm for a solitary predicate unless \NP=\RP.
However, any solitary predicate with the efficient advancement property can be detected in polynomial time.
\item 
The paper introduces two subclasses of solitary predicates: Those with the {\em antimonotone} advancement property
and with the {\em efficient rejection} property. We use these properties to detect solitary predicates
in \NC.
\end{itemize}
\section{Related Work}
Predicates on a distributive lattice whose detection models various combinatorial algorithms have been introduced in
\cite{DBLP:conf/spaa/Garg20}, \cite{DBLP:conf/sss/Garg21}, \cite{DBLP:conf/sss/GuptaK21}.
These papers study the class of predicates called lattice-linear predicates which is equivalent to being 
closed under the meet operation of the lattice. 
It has been shown that the lattice-linear predicate (LLP) algorithm solves many combinatorial optimization problems such as the shortest path problem, the stable marriage problem, and the market clearing price problem \cite{DBLP:conf/spaa/Garg20}. 
This method has been applied to other problems such as dynamic programming problems \cite{Garg:ICDCN22},
the housing allocation problem \cite{DBLP:conf/sss/Garg21}, the minimum spanning tree problem \cite{AlvGar22},
and generalizations of the stable matching problem \cite{DBLP:conf/icdcn/Garg23}.
In \cite{DBLP:conf/sss/GuptaK21}, Gupta and Kulkarni extend LLP algorithms for deriving self-stabilizing algorithms. 
In \cite{DBLP:journals/corr/abs-2302-07207}, Gupta and Kulkarni also
show that multiplication and modulo operations on natural numbers can be modeled
as LLP algorithms. In  \cite{gupta2023lattice}, these authors continue by showing lattice-linearity simplifies the analysis of robot coordination algorithms in an asynchronous setting, and improves upon the algorithm's complexity guarantees.

The above-mentioned work and the work in this paper are based on obtaining solutions to combinatorial optimization and constraint satisfaction problems through predicate detection algorithms.
In the context of distributed monitoring, the technique of predicate detection was introduced by Cooper and Marzullo \cite{CoopMarz:ConsDetGP} and Garg and Waldecker \cite{GarWal:wpdd}.
Regarding existing classes of predicates introduced for this setting, the detection of conjunctive predicates was examined in 
\cite{GargWald:WeakUnstable},
lattice-linear predicates were introduced in \cite{DBLP:journals/dc/ChaseG98}, and regular
predicates were introduced in \cite{MG:dcs01slice}.
Moreover, other classes specific to distributed computing settings, such as 
observer-independent predicates \cite{CDF95}, have also been investigated.

Note that equilevel predicates holding for multiple
elements on one level cannot be closed under taking meets on a lattice. Thus, the LLP algorithm is inapplicable for these predicates.
In contrast, when restricted to solitary predicates LLP algorithms can be applied. 
Thus, to advance this line of work we identify additional properties allowing parallel \NC~ algorithms to
detect solitary predicates.

\section{Equilevel Predicates}

Throughout, we consider problems defined on a distributive lattice representing the domain.
To ease the presentation of our definitions, we consider a specific representation.
By Birkhoff's theorem \cite{Birk1},
any finite distributive lattice $\CC$
can be generated by a finite poset $P$. Let that poset $P$ consist of $n$ chains, i.e.,
$P$ can be decomposed into $n$ chains $P_1, P_2, \ldots P_n$ of distinct elements.
Thus, we can assume $\CC$ is represented this way without loss of generality, and every element $G \in \CC$ can be viewed as an ideal (or, downset) of $P$.
Moreover, when examining parallel algorithms we imagine that there are $n$ processes where each process corresponds to exactly one of these chains.
In this way, we treat each element $G\in\CC$ as a global state in the parallel computation, and the largest element of the chain $P_i$ present in $G$ as the local state of process $i$.
It follows then that a process \emph{advances} a computation by augmenting the current global state with the next element in its respective chain.

Before continuing, let $G[i]$ denote the
number of elements in $G$ from $P_i$, while $|G|$ denotes the cardinality counting all elements in the ideal representing $G$.
Finally, $\bot$ and $\top$ correspond to the bottom and top elements of the lattice respectively.
Then, a predicate on $\CC$ is \emph{equilevel} if all the satisfying elements are on the same level of the lattice.
\remove{The level is given by the unique function $\rho:\CC\to\mathbb{N}$ such that $\rho$ applied to the bottom element equals one, $\rho$ is compatible with the ordering, i.e. $G < H \implies \rho(G) < \rho(H)$, and that $\rho$ recovers the covering relation of the lattice in the sense that $\rho(G) = \rho(H) - 1$ whenever $H$ covers $G$.
Then, the value $\rho(G)$ gives the level of $G$ in $\CC$.}
\begin{definition}[Equilevel Predicate]

A Boolean predicate $B$ is {\em {equilevel}} with respect to a lattice $\CC$
iff
$$\forall G, H \in \CC: B(G) \wedge B(H) \Rightarrow |G| = |H|.$$
\end{definition}
\remove{
\begin{figure}
\begin{center}
\begin{tikzpicture}[scale=0.5]
\draw (0,0) -- (-2,4) -- (0,8);
\draw (0,0) -- (2,4) -- (0,8);
\draw[color=red] (-1,4) circle (2pt);
\draw[color=red] (+1,4) circle (2pt);
\end{tikzpicture}
\caption{Search space with elements that satisfy the given predicate $B$. The satisfying solutions are shown in red. They are all at the same level.
 \label{fig:equilevel}
}
\end{center}
\end{figure}}
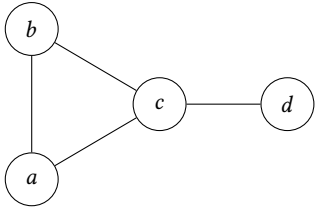
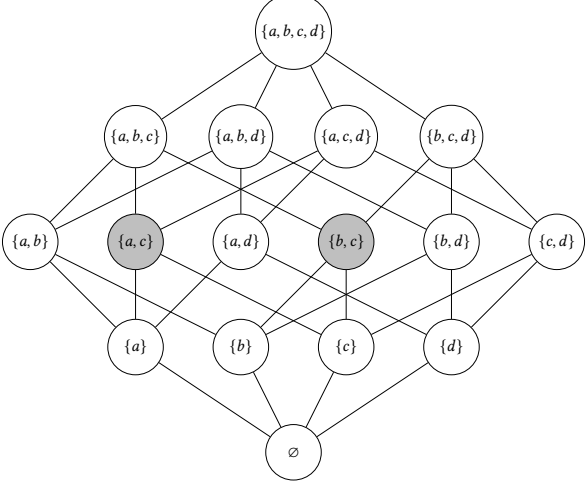
\begin{figure}
    \begin{subfigure}{0.45\textwidth}
  \centering
  \begin{tikzpicture}
    \tikzstyle{vertex}=[circle, draw, inner sep=2pt, minimum size=20pt]
    
    \node[vertex] (a) at (0,0) {$a$};
    \node[vertex] (b) at (0,2) {$b$};
    \node[vertex] (c) at (1.7,1) {$c$};
    \node[vertex] (d) at (3.4,1) {$d$};
    
    \draw (a) -- (c);
    \draw (a) -- (b);
    \draw (c) -- (d);
    \draw (b) -- (c);
  \end{tikzpicture}
  \caption{A graph $G = (\{a, b, c, d\}, \{(a, b), (b, c), (a, c), (c, d)\})$.}
\end{subfigure}
\begin{subfigure}{0.45\textwidth}
\centering
\scalebox{.7}{
    \begin{tikzpicture}
    \tikzstyle{vertex}=[circle, draw, inner sep=2pt, minimum size=30pt]
    \tikzstyle{filled}=[fill=lightgray] 
  
    \node[vertex] (empty) at (0,0) {$\varnothing$};
    \node[vertex] (a) at (-3,2) {$\{a\}$};
    \node[vertex] (b) at (-1,2) {$\{b\}$};
    \node[vertex] (c) at (1,2) {$\{c\}$};
    \node[vertex] (d) at (3,2) {$\{d\}$};
    \node[vertex] (ab) at (-5,4) {$\{a, b\}$};
    \node[vertex] [filled] (ac) at (-3,4) {$\{a, c\}$};
    \node[vertex] (ad) at (-1,4) {$\{a, d\}$};
    \node[vertex] [filled] (bc) at (1,4) {$\{b, c\}$};
    \node[vertex] (bd) at (3,4) {$\{b, d\}$};
    \node[vertex] (cd) at (5,4) {$\{c, d\}$};
    \node[vertex] (abc) at (-3,6) {$\{a, b, c\}$};
    \node[vertex] (abd) at (-1,6) {$\{a, b, d\}$};
    \node[vertex] (acd) at (1,6) {$\{a, c, d\}$};
    \node[vertex] (bcd) at (3,6) {$\{b, c, d\}$};
    \node[vertex] (abcd) at (0,8) {$\{a, b, c, d\}$};
  
    \draw (empty) -- (a);
    \draw (empty) -- (b);
    \draw (empty) -- (c);
    \draw (empty) -- (d);
    \draw (a) -- (ab);
    \draw (a) -- (ac);
    \draw (a) -- (ad);
    \draw (b) -- (ab);
    \draw (b) -- (bc);
    \draw (b) -- (bd);
    \draw (c) -- (ac);
    \draw (c) -- (bc);
    \draw (c) -- (cd);
    \draw (d) -- (ad);
    \draw (d) -- (bd);
    \draw (d) -- (cd);
    \draw (ab) -- (abc);
    \draw (ab) -- (abd);
    \draw (ac) -- (abc);
    \draw (ac) -- (acd);
    \draw (ad) -- (abd);
    \draw (ad) -- (acd);
    \draw (bc) -- (abc);
    \draw (bc) -- (bcd);
    \draw (bd) -- (abd);
    \draw (bd) -- (bcd);
    \draw (cd) -- (acd);
    \draw (cd) -- (bcd);
    \draw (abc) -- (abcd);
    \draw (abd) -- (abcd);
    \draw (acd) -- (abcd);
    \draw (bcd) -- (abcd);
\end{tikzpicture}}
\caption{Hasse diagram of the vertex powerset, ordered by inclusion.}
\end{subfigure}
    \caption{Above is the example of an equilevel predicate in the context of the minimum vertex cover problem. Examine the given graph (top). The minimum vertex covers are of cardinality 2, given by $\{a, c\}$ and $\{b, c\}$. This gives an equilevel predicate over the inclusion order (bottom) such that the satisfying elements (shaded) are all on level 2.}
    \label{fig:equilevel}
\end{figure}
Fig. \ref{fig:equilevel} shows a search space (modeled using a distributive lattice) and 
a predicate and its set of satisfying solutions which are all on the same level of the lattice.
We now define the notion of ``detecting'' a predicate.

\begin{definition}[Predicate Detection]
Given a poset $P$ generating a lattice 
$\CC$, and a Boolean predicate $B$, if $B$ is true on some element $G$ of lattice $\CC$, then decide ``yes.''
Otherwise, decide ``no.''
\end{definition}

For example, consider the minimum vertex cover problem in an undirected graph $(V, E)$.  A subset of vertices $V'$ satisfies $B$ if $V'$ is the vertex cover of minimum size. Here, the lattice we consider is the
Boolean lattice of the vertex set. It is clear that $B$ is an equilevel predicate since there may be multiple minimum vertex covers and they must all be of the same size. Then, we observe that detecting an equilevel
predicate is \NP-complete in general as the minimum vertex-cover decision problem
is well-known to be \NP-complete.

As another example of an equilevel predicate,
suppose that we have a bipartite graph $(L, R, E)$ and we are interested in maximum-sized subsets
of $L$ that can be matched. The size of maximum-sized subsets of $L$ that can be matched to elements in $R$
is constant, even though there may be multiple matched sets in $L$. 
For example, consider $L = \{m_1, m_2, m_3\}$, $R = \{w_1, w_2, w_3\}$ and $E = \{ (m_1, w_1), (m_2, w_1), (m_3, w_3) \}$.
Then, we can match $\{m_1, m_3\}$ or $\{m_2, m_3\}$. Both sets are of size $2$.

Equilevel predicates occur in numerous other contexts. 
Consider the problem of finding a minimum weight spanning tree of a connected undirected graph $(V, E)$ with $n$ vertices.
If the edge weights are not unique, then there may be multiple minimum-spanning trees with
equal weight. 
Let the predicate $B$ be ``the set of edges form a spanning tree with minimum weight.''
It is clear that all sets that satisfy the predicate have $n-1$ edges. Thus, in the inclusion lattice over the graph's edges, the sets satisfying the
predicate are all at the same level.

Equilevel predicates are closed under conjunction. If predicates $B_1$ and $B_2$ are true for different levels, then the conjunction is 
false for all elements of the distributive lattice and the predicate $B_1 \wedge B_2$ is trivially equilevel.
Otherwise, if they are true for the same level, then $B_1 \wedge B_2$ is also true only at that level (or possibly, always false).
However, equilevel predicates are not closed under disjunction or complement because these operations may introduce elements satisfying the new predicate at many different levels of the lattice. 

As we have observed, since \NP-complete problems such as vertex cover can be modeled
as detecting an equilevel predicate, the problem of detecting an arbitrary equilevel predicate is NP-hard.
As another example, consider the problem of detecting a Hamiltonian path in an undirected graph.
If we consider the Boolean lattice of all edges, then we only need to consider the subset of edges at level
$n-1$ on a graph with $n$ vertices. Thus, the problem of Hamiltonian path can also be modeled as
detecting an equilevel predicate. 

We now show that some problems in distributed computation with efficient algorithms are equilevel predicates
when slightly generalized.
For example, consider \emph{conjunctive predicates} which are given by the Boolean formula
$l_1 \wedge l_2 \wedge \ldots l_n$ where each $l_i$ is a function only of state local to the $i^\text{th}$ process. 
For this class, there is an efficient algorithm for detecting this predicate in a distributed computation
\cite{GargWald:WeakUnstable}. However, if we generalize the problem by asking for
the conjunctive predicate on a particular level of the lattice, then the problem
can be modeled as an equilevel predicate.
We now show that, in general, asking for a conjunctive predicate on a particular level
is \NP-complete.
\begin{theorem}
Given a distributed computation, deciding
whether there exists a global state with $k$ events satisfying a
given conjunctive predicate is \NP-complete.
\end{theorem}
\begin{proof}
We first show that the problem is in \NP. The global state itself provides a succinct certificate. We can check that all local predicates are true in that global state
and that the global state is at level $k$.

For hardness, we use the subset sum problem. Given a subset problem on $n$ positive integers, $x_1, x_2, \ldots, x_n$ with the requirement to choose a subset that adds up to $k$, we create a computation on $n$ processes as
follows. Each process $P_i$ has $x_i$ events. The local predicate on $P_i$ is true
initially and also after it has executed $x_i$ events. Thus, the local predicate
is true at each process exactly twice. The problem asks us if there is a 
global state with $k$ events in which all local predicates are true. Such a global 
state, if it exists, would choose for every process either the initial local state or the final
local state. All the final states that are chosen correspond to the numbers that have been 
chosen.

To avoid the expansion of the numbers in binary to unary construction, we encode the 
representation of events on a process as follows: Since the conjunctive predicates
can only be true when the local predicates are true, we keep only local states which satisfy their corresponding local predicate and store the number of local events executed so far with them. This leaves two local states at each process: The initial state with zero events
executed until that point, and the final state of the process with the number of events equal to $x_i$ for the $i^\text{th}$ process. Now, the construction
of the computation from the subset sum problem is polynomial in the size of the input.
\end{proof}

In what follows, we restrict our attention to a class of equilevel predicates to one which can be detected efficiently.

\section{Equilevel Predicates with The Helpful Property}
 
One class of equilevel predicates that we can efficiently detect is the 
set of predicates that satisfy the {\em helpful} property.
An equilevel predicate satisfies the helpful property if, whenever it is false on
a global state $G$, we are guaranteed to be able to compute in polynomial time the indices of the local states which are {\em helpful}. Unless the global state $G$
is advanced on {\em at least one} helpful local state the predicate
can never become true.

Before, we define the helpful property, we introduce the
notation $B(\CC)$ to mean that there exists $W \in \CC$ such that $B(W)$.
Now, we define the helpful property.
 
\begin{definition}[Helpful Property]
A Boolean predicate $B$ has the {\em helpful} property with a polynomial time algorithm $\mathcal{A}$ on a finite distributive lattice $\CC$
if, for all $G \in \CC$,
\[ \neg B(G) \Rightarrow 
\mathcal{A}(G)\neq\varnothing \wedge\forall i \in \mathcal{A}(G): helpful(i, G, B).\]
where 
$helpful(i,G,B)$ is defined as
\[ B(\CC) \Rightarrow \exists H > G: (H[i] > G[i]) \wedge B(H).\]
\end{definition}

The definition states that whenever the predicate is not true in a global state $G$, the polynomial time
algorithm $\mathcal{A}$ can return a nonempty set of indices $\mathcal{A}(G)$ such that advancing on any of the corresponding local states
can be used to detect the predicate. 
The $i^\text{th}$ local state is helpful in the global state $G$
if we can advance the global state $G$ on the index $i$ whenever $\neg B(G)$ holds without the risk of missing a satisfying global state $H$.
If there is no global state that satisfies $B$, then $helpful(i,G,B)$ holds vacuously.

The helpful property allows us to efficiently detect
an equilevel predicate satisfying the property.
In all our examples, we also require the algorithm $\mathcal{A}$ to be {\em online}
in the sense that the algorithm $\mathcal{A}$ has access to the poset only on states up to $G$.
Therefore, the computation of helpful indices is based only on the poset of states less than or equal to $G$.
This allows us to use a sequential procedure, given in Algorithm~\ref{fig:alg-equi}, to advance on helpful indices towards a satisfying global state if one exists.

The Algorithm~\ref{fig:alg-equi} works as follows. It starts with the 
global state $G$ initialized to the least global state corresponding to the bottom
element of the lattice $\CC$. If the predicate $B$ is true at the initial global state, then
we are done and we can return $G$. If not, we first check whether $G$ corresponds to the
top element of the lattice. If $G$ is the top element of the lattice, we can safely return that
there is no element in $\CC$ that satisfies $B$. Otherwise, we advance on the index $i$
where $helpful(i,G,B)$ holds.

\begin{algorithm}
\SetAlgoRefName{Equilevel}
{\bf function} GetSatisfying($B$: predicate, $\CC$: Lattice)\\
\h {\bf var} $\CGG$: array[$1 \ldots m$] of int initially $\forall i: \CGG[i] = 0$;\\
 \h {\bf while} $\neg B(\CGG)${\bf do}\\

\h \h {\bf if} ($\CGG$ is the top element of $\CC$)\\
\h \h \h then {\bf return} null;\\
 \h \h {\bf else} \\ 
\h  \h \h  $\CGG[i] := \CGG[i]+1$ where $helpful(i, G, B)$;\\
      \h {\bf endwhile};\\
  \h     {\bf return} $\CGG$; // an optimal solution \\
\caption{A Sequential Algorithm  to find a state satisfying $B$\label{fig:alg-equi}}
\end{algorithm}

Theorem \ref{thm:equi-helpful} gives the correctness
of Algorithm~\ref{fig:alg-equi}.

\begin{theorem} \label{thm:equi-helpful}
    Let $B$ be any equilevel predicate that satisfies a
    helpful property with a polynomial time algorithm $\mathcal{A}$.
    Then, $B$ can be detected online in polynomial time.
\end{theorem}
\begin{proof}
    Algorithm~\ref{fig:alg-equi} starts with the least element in the lattice $\CC$. The {\em while} loop can execute at most $mn$ times where $n$ is the number of chains and $m$ is the maximum height of any chain. The algorithm $\mathcal{A}$ to compute $helpful$ has polynomial complexity. Hence, Algorithm~\ref{fig:alg-equi} also has polynomial complexity. The algorithm $\mathcal{A}$ has access to the lattice $\CC$ only 
    up to $\CGG$.

    Since $B$ is an equilevel predicate, the algorithm is guaranteed to return a minimal state that 
    satisfies $B$.
\end{proof}

By {\em online} detection of the predicate, we mean that the algorithm has access to the lattice 
only up to the current global state $\CGG$. For example, when we consider the stable marriage problem, the
vector $\CGG$ denotes the current assignment of men to women. The algorithm needs to make its decisions
based only on the information on men preferences and women preferences up to that point.
This interpretation is different from the online algorithm in the standard algorithm literature, where
the information about the problem comes incrementally.

We show detection of equilevel predicates through the following examples:

{\em Example 1: Maximum Cardinality Bipartite Matching}.
Consider the problem of bipartite matching in a graph $(L, R, E)$.
We are interested in finding a subset $V \subseteq L$ of maximum size such that
$V$ can be matched with vertices in $R$.
We
use the Boolean lattice on the vertices $L$. 
A vertex set $V \subseteq L$ satisfies $B$ if
$|V|$ is the maximum number of vertices that can be matched.
It is clear that $B$ is an equilevel predicate.
Suppose that we have a set of vertices $V \subseteq L$ which is not of the maximum size.
Then, we can find a set of helpful set of vertices $W$ such that any 
element of $W$ can be added to $V$.
It can be checked efficiently that $V$ does not satisfy $B$
(for example, by checking if there exists an alternating path from a vertex in $L-V$ to an unmatched vertex in $R$
\cite{CLR}).
The algorithm to find helpful vertices is simple: any vertex in $L-V$ that has an alternating
path to an uncovered vertex in $R$ is a helpful vertex. For efficiency, the algorithm may maintain the set of matched vertices in $R$.
One can use Algorithm \ref{fig:alg-equi} to find the largest set in $L$ with matching in $R$.
The predicate $helpful(i, G, B)$ holds whenever the vertex $i$ can be added to the current set of matched
vertices $G$.
We note here that \cite{fenner2019deterministic} discusses parallel algorithms for perfect matching
in a bipartite graph. Our goal is to view the problem from the perspective of detecting a global condition.
\qedsymbol

{\em Example 2: Minimum Spanning Tree in an Undirected Graph}.
Let $B$ be true on a set of edges of an undirected graph if the edge set forms a spanning tree of the graph.
Given a connected undirected graph with $n$ vertices, it is well-known that
any edge set that is acyclic and has $n-1$ edges satisfies $B$. 
This predicate is equilevel because all the elements that satisfy the predicate $B$ are
at the level $n-1$ in the Boolean lattice of all edges.
Given a set of edges, one can efficiently compute whether the set is acyclic. 
Hence, one can use Algorithm \ref{fig:alg-equi} to find a spanning tree as follows:
The predicate $helpful(i, G, B)$ holds whenever the edge $e_i$ does not form a cycle
with edges in $G$, so at each step Algorithm~\ref{fig:alg-equi} adds any edge which does not make a cycle in the current solution.
\qedsymbol

{\em Example 3: Basis Set of Vectors}.
Suppose that 
we are given a set of vectors $S$ of the same dimension. Our goal is to compute a basis, i.e. find a subset of the biggest size
containing linearly independent vectors. The predicate $B$ is true on a set $G$
if all the vectors in $G$ are linearly independent, and there does not exist
any vector in $S-G$ that is linearly independent with vectors in $G$.
Here $helpful(i,G,B)$ holds if the vector $v_i$ is independent of all the vectors in $G$.
\qedsymbol

Lattice-linear predicates \cite{DBLP:conf/spaa/Garg20} are associated with the concept of \emph{forbidden} local states.
In particular, for any $G$ which does not satisfy a lattice-linear predicate $B$, there exists at least one local state $s$ such that no global state that satisfies $B$ can have the same local state $s$.
This means that one can advance on all forbidden local states of a lattice-linear predicate unconditionally.
In particular, if there are multiple local states that are forbidden then the global state can be advanced on all of them in parallel as forbidden local states remain forbidden until they are advanced.
In contrast, when it comes to equilevel predicates, advancing on one helpful
local state may make the local state on some other process ``unhelpful.''  For example, in computing a spanning tree we may be able to add
one of edge $a$ or edge $b$ to our current solution, but adding both of them could create a cycle.

An equilevel predicate may not have the helpful property as shown next.
\begin{theorem}
 There exist equilevel predicates that do not have any helpful property unless \P = \NP.
\end{theorem}
\begin{proof}
  We have earlier shown that the problem of finding the vertex cover of the least size is
  an equilevel predicate. If there exists any helpful property for this problem, we 
  would have a polynomial time algorithm to find the vertex cover.
  Similarly, detecting a conjunctive predicate at a particular level of the lattice of global states also cannot have the helpful property unless \P = \NP.
\end{proof}

\section{Equilevel Predicates with the Independently Helpful Property}

We now define a stronger version of equilevel predicates in which 
at every step, one evaluates which indices are independently helpful. Two helpful indices are
{\em independently helpful} if 
it is okay to advance on both of them in parallel. 
This definition allows us to design parallel algorithms for equilevel predicates.

Formally,

\begin{definition}[Independently Helpful Property]
 A Boolean predicate $B$ has an independently helpful property with the \NC~algorithm
 $\mathcal{A}$ on a finite distributive lattice $\CC$
 if for all $G \in \CC$,
 \[ \neg B(G) \Rightarrow \mathcal{A}(G)\neq\varnothing \wedge independently\text{-}Helpful(\mathcal{A}(G), G,  B).\]
where 
$independently\text{-}Helpful(J, G, B)$ is defined as
\[ B(\CC) \Rightarrow \exists H > G: \forall i \in J : (H[i] > G[i]) \wedge B(H).\]
such that with a polylogarithmic number of parallel advancements, the predicate is detected.
\end{definition}

Observe that the predicate $independently\text{-}Helpful(J, G, B)$ implies the predicate
$helpful(i, G, B)$ for all $i \in J$. In this definition, we require that
there exists $H$ that is bigger than $G$ in {\em all} of the components $i \in J$. The definition of $helpful(i, G, B)$ required $H$ to be bigger than $G$ only in the
component $i$. Hence, the algorithm can advance in all of the components in $J$ 
in parallel. Secondly, we require that we either reach the top of the lattice or find
a global state that satisfies the predicate in at most polylogarithmic number 
of parallel advancements. 
Note the requirements on the algorithm $\mathcal{A}$: We require it to be of
polylogarithmic depth complexity so that we can detect the predicate in \NC.

Any equilevel predicate with independently helpful property can be detected in \NC~by using the
Algorithm \ref{fig:alg-equi-indep}.

\begin{theorem} \label{thm:equi-independent}
    Let $B$ be any equilevel predicate that satisfies an
    independently helpful property with an \NC~algorithm $\mathcal{A}$.
    Then, $B$ can be detected in \NC.
\end{theorem}
\begin{proof}
    Algorithm~\ref{fig:alg-equi-indep} starts with the least element in the lattice $\CC$. In each step of the {\em while} loop, it makes one call of independently helpful and advances on
    all indices $i \in J$ with $independently\text{-}Helpful(J,G,B)$. If there can be at most polylogarithmic number of advancements, then
    the algorithm will detect $B$ in \NC.
\end{proof}

{\em Example 1}: For the problem of minimum-spanning tree, any set of edges that do not form a cycle with the currently chosen edges is independently helpful. Observe that there may be multiple sets of independently helpful states. We only require the algorithm to 
return any such set. How do we find an independently helpful set in parallel?
At every stage in the algorithm, there is a set of connected components.
Each connected component chooses a single outgoing edge such that the edges chosen 
do not form any cycle. This set of edges can be added 
in parallel to reach the next stage of the algorithm.
The number of connected components decreases by at least a factor of two after every such step.
Hence, we would need to make at most polylogarithmic advancements.
This leads us to Borůvka's parallel algorithm for the minimum spanning tree \cite{Nesetril.2001}.
\qedsymbol


\begin{algorithm}
\SetAlgoRefName{Equilevel-Independence}
vector {\bf function} GetSatisfying($B$: predicate, $\CC$: Lattice)\\
{\bf var} $\CGG$: array[$1 \ldots m$] of int initially $\forall i: \CGG[i] = 0$;\\
 {\bf while} $\neg B(\CGG)${\bf do}\\
 \h {\bf if} ($\CGG$ is the top element of $\CC$)\\
 \h \h then {\bf return} null;\\
 \h {\bf else} \\ 
 \h \h {\bf forall} $i \in J$ where $J = \mathcal{A}(G)$\\
 \h \h \h and $independently\text{-}Helpful(J,G, B)$ 
 {\bf in parallel} do\\
\h \h  \h \h  $\CGG[i] := \CGG[i]+1$;\\
       {\bf endwhile};\\
    {\bf return} $\CGG$; // an optimal solution \\
\caption{A Parallel \NC~Algorithm  to find a state satisfying $B$\label{fig:alg-equi-indep}}
\end{algorithm}

\section{Solitary Predicates}

A {\em solitary} predicate is one which is either false on all the elements of the lattice, or is true
on a single element in the lattice. 
\begin{definition}[Solitary Predicate]
A Boolean predicate $B$ is {\em {solitary}} with respect to a lattice $\CC$
iff
$$\forall G, H \in \CC: B(G) \wedge B(H) \Rightarrow (G=H).$$
\end{definition}
The above definition states that if the predicate is true for two elements $G$ and $H$ in the lattice, then
$G$ must be equal to $H$. This definition also includes the empty predicate that is not true in any element
of the lattice. 

Solitary predicates are closely related to the problem of Unique SAT.
Unique SAT (USAT) \cite{DBLP:conf/stoc/ValiantV85} asks for the satisfiability of a Boolean expression given the promise that it has either a single satisfying assignment or none. 
Let $x_1, x_2, \ldots, x_n$ be $n$ Boolean variables. Let $B$ be a Boolean expression on these variables. The USAT problem requires the algorithm to return $1$, if the Boolean expression 
has a unique satisfying assignment, $0$, if the Boolean expression is not satisfiable
and either $0$ or $1$, if it has multiple satisfying assignments.
Valiant and Vazirani \cite{DBLP:conf/stoc/ValiantV85} have shown the following result.

\begin{theorem}
({\bf Valiant-Vazirani}\cite{DBLP:conf/stoc/ValiantV85})  If there exists a randomized polynomial-time algorithm for solving instances of SAT with at most one 
satisfying assignment, then \NP=\RP.
\end{theorem}

As a direct application, we get the following result by using the
construction in \cite{DBLP:journals/ipl/KashyapG05} as shown below.
\begin{theorem}
Given any finite distributive lattice $\CC$ generated from a poset $P$,
and a solitary predicate $B$, there does not exist a randomized polynomial time algorithm
to detect $B$ unless \NP=\RP.
\end{theorem}
\begin{proof}
Suppose, if possible, there exists a randomized polynomial time algorithm $\mathcal{A}$ to detect a solitary predicate.
Given any instance of USAT on $n$ variables $x_1, x_2, \ldots, x_n$, we construct a poset with $n$ events on $n$ processes, $P_1, P_2, \ldots, P_n$ with one event per process. Process $P_i$
hosts the variable $x_i$ which is initially set to false. When the event is executed, the variable
is set to true. The distributive lattice generated from this poset has $2^n$ elements --- one for
every truth assignment. If the given instance of USAT has zero satisfying assignments, then
$\mathcal{A}$ is guaranteed to return $false$. If the given instance of USAT has exactly one satisfying 
assignment, then $\mathcal{A}$ is guaranteed to return $true$ with a satisfying assignment. Otherwise,
we do not care what $\mathcal{A}$ returns.
Therefore, any randomized polynomial time algorithm to detect a solitary predicate can be used to solve USAT
in polynomial time.
\end{proof}

\section{Solitary Predicates with Efficient Advancement Property}

Although detecting a general solitary predicate is hard, we can detect the predicate $B$,
whenever it satisfies the efficient advancement property \cite{DBLP:conf/spaa/Garg20}. The efficient advancement
property requires an efficient algorithm $\mathcal{A}$ such that whenever the predicate is false
on an element $G \in  L$, the algorithm returns an index $i$ such that the predicate
is false for all $H \geq G$ such that $H[i]$ equals $G[i]$.

\begin{definition}[Efficient Advancement Property \cite{DBLP:conf/spaa/Garg20}]
A Boolean predicate $B$ has the efficient advancement property with the polytime algorithm $\mathcal{A}$ with respect to a finite distributive lattice $\CC$ generated from a poset $P$
if 
$$\forall G \in \CC: \neg B(G) \Rightarrow \exists i \in \mathcal{A}(G): forbidden(i, G, B)$$
where $forbidden(i,G,B)$ is defined as
\[ \forall H \geq G: (H[i] = G[i]) \Rightarrow \neg B(H).\]
\end{definition}

The above definition states that whenever $B$ is false in $G$, the algorithm $\mathcal{A}(G)$ returns
a set of indices such that any global state $H$ greater than $G$ that matches $G$ on any returned index also has $B$ false. 
We will restrict ourselves to online algorithms, i.e., the algorithm
$\mathcal{A}$ will only have access to those global states preceding $G$.


Any problem that can be modeled using a solitary predicate with the efficient advancement property
can be solved in polynomial time. We restate the result from \cite{DBLP:conf/spaa/Garg20}
in terms of solitary predicates.
\begin{theorem}\label{thm:singalg1} \cite{DBLP:conf/spaa/Garg20}
Let $B$ be a solitary predicate with the efficient advancement property on a distributive lattice $\CC$.
Then $B$ can be detected in polynomial time.
\end{theorem}

It is easy to develop a parallel algorithm to detect a solitary predicate
with the efficient advancement property.

\begin{algorithm}
\SetAlgoRefName{Solitary}
{\bf var} $\CGG$: element of $\CC$ initially $\bot$;\\
{\bf while} $\neg B(\CGG)$ do\\
\h  {\bf forall} $i$: forbidden($i$, $\CGG$) do {\bf in parallel}\\
\h \h {\bf if} $\CGG$ cannot be advanced on $i$ then return false;\\
\h \h {\bf else} advance $\CGG$ on $i$;\\
\caption{An Online Parallel Algorithm for detecting Solitary predicates with the efficient advancement property. \label{fig:alg-solitary}}
\end{algorithm}

Now consider lattice-linear predicates that are true on multiple elements in the lattice.
Let $B$ be any lattice-linear predicate. We know that $B$ is closed under the meet operation of the lattice \cite{DBLP:conf/spaa/Garg20}.
Suppose that $B$ becomes true in the lattice and the least element that satisfies $B$ is $G_{min}$.
We derive another predicate $B'$ from $B$ that holds only for the element $G_{min}$ whenever it exists.
Thus, $B'$ holds for $G$ iff $B(G)$ and for all elements $H$, if $B(H)$, then $G \leq H$. Therefore,
$B'(G)$ is false for all other elements besides $G_{min}$.

\begin{theorem}
Let $B$ be any lattice-linear predicate. Let $B'(G)$ be defined as follows:
\[ B'(G) \equiv B(G) \wedge  (\forall H: B(H) \Rightarrow G \leq H). \]
Then, $B'$ is a solitary predicate.
\end{theorem}
\begin{proof}
First, consider the case when $B$ does not hold for any element in the lattice. This implies that $B'$ also does not 
hold for any element in the lattice and is, therefore, solitary.
Now, suppose that $B$ is true in the lattice. Since $B$ is an LLP predicate, it is
closed under meets. Therefore, there is the smallest element, $G_{min}$,
in the lattice such that $B(G_{min})$ holds. 
It is readily verified that $B'(G)$ holds iff $G=G_{min}$.
\end{proof}

Thus, all the instances of lattice-linear predicates provide us examples of solitary predicates when we focus on the least element.

{\em Example 1: Man-Optimal Stable Marriage.}
In this problem, we are given as input $n$ men and $n$ women. We are also given 
 a list of men preferences as $mpref$ where $mpref[i][k]$ denotes $k^{th}$ top choice of man $i$.
 The women preferences are more convenient to express as a $rank$ array where $rank[i][j]$ is the rank
 of man $j$ by woman $i$. A matching between man and woman is stable if there is no {\em blocking pair}, 
 i.e., a pair of woman and man such that they are not matched and prefer each other to their spouses.
 The underlying lattice for this example is
 the set of all $n$ dimensional vectors of $1..n$.
 We let $G[i]$ be the choice number that man $i$ has proposed to. Initially, $G[i]$ is $1$ for all men.

If we now focus on the man-optimal stable marriage, then the predicate ``the assignment is a man-optimal stable marriage''
is a solitary predicate.
The predicate $B_{stableMarriage}$ for the stable marriage is given by,
\[B_{stableMarriage} \equiv \forall j: \neg forbidden(j, G)\]
where
 $forbidden(j, G)$ is defined as
 $$(\exists i: \exists k \leq G[i]: (z = mpref[i][k])  \wedge (rank[z][i] < rank[z][j])),$$
 with $z = mpref[j][G[j]]$.

The predicate says that a marriage given by the vector $G$ is stable if none of
its index $j$ is forbidden. The index $j$ is forbidden if the woman $z$ corresponding 
to man $j$'s preference in $G[j]$ is also equal to the preference of man $i$ in $G$, or
a global state before $G$, and the woman $z$ prefers $i$ to $j$.

We now define the predicate for the man-optimal stable marriage,  $B_{mosm}(G)$ as
\[ B_{stableMarriage}(G) \wedge  (\forall H: B_{stableMarriage}(H) \Rightarrow G \leq H). \]

By definition, it is clear that $B_{mosm}$ is a solitary predicate. Given a lattice $\CC$,
we can search for the element satisfying $B_{mosm}$ by searching for $B_{stableMarriage}$.
We can use Algorithm \ref{fig:alg-solitary} for this purpose.
\qedsymbol

{\em Example 2: Housing Allocation Problem.}
As another example of a solitary predicate, consider the housing allocation problem with $n$ agents
and $n$ houses proposed by Shapley and Scarf \cite{shapley1974cores}.
The housing market is a matching problem with one-sided preferences.
 Each agent $a_i$ initially owns a house $h_i$ for $i \in \{1,n\}$ and
has a completely ranked list of houses.
The list of preferences of the agents is given by $pref[i][k]$ which specifies the $k^{th}$ preference of the
agent $i$. 
The goal is to come up with an optimal house allocation such that
each agent has a house and no subset of agents can improve the satisfaction of agents in this
subset by exchanging houses within the subset. 

It is well-known that the housing market has a unique solution, called the
{\em core} of the game.
We model the housing market problem as that of predicate detection in a computation \cite{DBLP:conf/sss/Garg21}.
Each agent proposes to houses in the decreasing order of preferences.
These proposals are considered as events executed by $n$ processes representing the agents.
Thus, we have $n$ events per process. 
The global state corresponds to the number of proposals made by each of the agents.
Let $G[i]$ be the number of proposals made by the agent $i$. We assume that
in the initial state, every agent has made his first proposal. Thus, the initial global state
$G = [1,1,..,1]$.
We extend the notation of indexing to subsets $J \subseteq [n]$ such that
$G[J]$ corresponds to the subvector given by indices in $J$.
A global state $G$ satisfies $matching$ if every agent proposes a different house.
We generalize $matching$ to refer to a subset of agents rather than the entire
set. Let $J \subseteq [n]$. Then,
$submatching(G, J)$  iff the houses proposed by agents in $J$
is a permutation of indices in $J$.
Intuitively, if $submatching(G, J)$ holds, then all agents in $J$ can exchange houses within the subset $J$.
For all $G$, there always exists a nonempty $J$ such that $submatching(G, J)$.
Let $B_{housing}(G)$ be defined as $G$ is a matching and 
\[ (\forall F <G: \forall J \subseteq [n]: submatching(F, J) \Rightarrow F[J] = G[J]). \]

It is easy to show that $B_{housing}(G)$ is a solitary predicate with an efficient advancement property.
Hence, the housing allocation problem also
can be modeled and solved using the solitary predicates with the efficient advancement
property. An agent $i$ is forbidden in the global state $G$ if the agent wishes a house that is part of the
submatching in $G$.
\qedsymbol

We briefly discuss detection of predicates when the search starts from the top of the lattice in addition to the bottom of the lattice. 
Many predicates, such as $B_{stableMarriage}$ and conjunctive predicates, satisfy not only the efficient advancement property but also its dual. Equivalently, the set of elements satisfying these predicates are closed not only
for the meet operation but also for the join operation. If we are okay with returning either
of the elements as our final answer, then searching for any of the element in parallel can speed
up the algorithm by a factor of the height of the lattice.
The dual of the efficient advancement property can formally be defined as \cite{MG:dcs01slice} follows.

\begin{definition}[dual of Efficient Advancement Property]
A Boolean predicate $B$ has the dual of efficient advancement property with the polytime algorithm $\mathcal{A}$ with respect to a finite distributive lattice $\CC$ generated from a poset $P$
if 
$$\forall G \in \CC: \neg B(G) \Rightarrow \exists i \in \mathcal{A}(G): dual\text{-}forbidden(i, G, B)$$
where $dual\text{-}forbidden(i,G,B)$ is defined as
 \[ \forall H \leq G: (H[i] = G[i]) \Rightarrow \neg B(H).\]
\end{definition}

The above definition states that whenever $B$ is false in $G$, the algorithm $\mathcal{A}(G)$ returns
an index such that any global state $H$ less than $G$ that matches $G$ on that index also has $B$ false.
When the efficient advancement as well as its dual are true, one can search for the satisfying element starting from both the bottom and the top of the lattice as shown in Algorithm \ref{fig:alg-solitary2}.

\begin{algorithm}
\SetAlgoRefName{Solitary2}
{\bf var} $\CGG$: element of $\CC$ initially $\bot$ (the bottom element of $\CC$);\\
\h $Z$: element of $\CC$ initially $\top$ (the top element of $\CC$);\\
\h {\bf while} $\neg B(\CGG)$  and $\neg B(Z)$ do\\
\h \h  {\bf forall} $i$: forbidden($i$, $\CGG$) do {\bf in parallel}\\
\h \h \h if $\CGG$ cannot be advanced on $i$ then return false;\\
\h \h \h else advance $\CGG$ on $i$;\\
\h \h  {\bf forall} $j$: dual-forbidden($j$, $Z$) do {\bf in parallel}\\
\h \h \h if $Z$ cannot be retreated on $j$ then return false;\\
\h \h \h else retreat $Z$ on $j$;\\
\h {\bf endwhile};\\
\h if $B(\CGG)$ then {\bf return} $\CGG$ else return $Z$; // an optimal solution \\
\caption{A Parallel Algorithm for detecting Predicates with efficient advancement property and its dual\label{fig:alg-solitary2}.}
\end{algorithm}

Applying the idea to the stable marriage problem, one can search for the stable marriage starting from
the top choices for all men (the $\bot$ element) of the lattice in parallel with the bottom choices for all men (the $\top$)
element of the lattice. This algorithm will traverse the distance in a lattice given by the minimum of the distance
of a stable marriage from the top or the bottom. Hence, depending upon the distance
of man-optimal and man-pessimal stable marriage, it will return the stable marriage that is closer to the
bottom or the top of the lattice.

Similarly, consider the problem of detecting conjunctive predicates
which are also closed under the join operation of the lattice. By using the Algorithm \ref{fig:alg-solitary2},
we again get an algorithm to detect a conjunctive predicate that will traverse the lattice given by the 
minimum of the distance of a satisfying global state from the top or the bottom of the lattice.
In the above analysis, we are ignoring the factor of $2$ penalty that we incur because
we run the algorithm both from the bottom and the top of the lattice.

\section{Solitary Predicates with \NC~Advancement Properties}

We now define a special case of the efficient advancement property.
A predicate $B$ has {\em the \NC~advancement} property,
if (1) there exists an \NC~algorithm to detect whether an index is forbidden and
(2) starting from the initial state and always advancing all the forbidden indices
the algorithm either reaches a satisfying state or the top element of the lattice in 
the polylogarithmic number of steps in the size of the input.
Clearly, any predicate that has the efficient \NC~advancement property
can be detected in parallel in \NC~time with the Algorithm~\ref{fig:alg-solitary}.

We now give several examples guaranteeing the \NC~advancement property.
We say that a predicate $B$
has an {\em antimonotone advancement} property if once an index $i$ is not forbidden in $G$, it
stays not forbidden for all $H \geq G$ such that $H[i]$ equals $G[i]$.
Formally,
\begin{definition}[Antimonotone Advancement Property]
A Boolean predicate $B$ with the advancement property (w.r.t. a poset $P$) is
{\em antimonotone}
if
\begin{equation*}
    \begin{split}
        \forall G \in \CC: &\neg forbidden(i, G, B) \Rightarrow \\
        &\forall H \geq G \wedge (H[i] = G[i]): \neg forbidden(i, H, B).
    \end{split}
\end{equation*}
\end{definition}

Consider the problem of the minimum spanning tree in an undirected graph when all edges have {\em unique weights}. If the graph is connected, then there is a unique minimum spanning tree.
To find this spanning tree, we consider the Boolean lattice formed from all edges. We assume that edges are given to us in the increasing order. 
We define the predicate $B$ on a subset of edges $G$ as true whenever $G$ 
forms the minimum spanning tree. 
We use the binary representation of $G$: 
The variable $G[i]$ equals $1$ iff the $i^{th}$
edge is part of the unique minimum spanning tree.
The predicate $B$ is solitary because either there is a unique
minimum spanning tree or no spanning tree in such a weighted undirected graph.
Furthermore, the predicate $B$ satisfies {\em antimonotone advancement} property. 
If $i^{th}$ edge is not part of the minimum spanning tree, then $G[i]$ is not forbidden and and it will stay not forbidden
even when other edges are included as part of $G$.

We now show that any solitary predicate with Antimonotone advancement property
on a poset with height polylogarithmic in $n$ can be detected in \NC.
\begin{theorem}\label{thm:antimono}
Let $P$ be any poset and $B$ be an antimonotone advancement predicate that can be checked in
\NC.
If the height of the poset is $O(\log^k(n))$, for any constant $k$, then
the Algorithm \ref{fig:alg-solitary} is in \NC.
\end{theorem}
\begin{proof}
 The {\em forall} statement in the algorithm \ref{fig:alg-solitary} can run at most $O(\log^k(n))$ time because if a process
 is forbidden then it must advance. All processes can advance at most $O(\log^k(n))$ times because the height 
 of the poset is $O(\log^k(n))$. We are exploiting the fact that all processes which can advance do so in parallel.
 If a process is not forbidden, then it stays not forbidden due to the antimonotonicity property.
 Since checking for the property is in \NC, the entire algorithm is in \NC.
\end{proof}

When we apply this algorithm to the minimum spanning tree problem, we get 
the \NC~algorithm~\ref{fig:alg-kruskal-llp}. The algorithm assumes that edges of the graph are
presented in the sorted order similar to Kruskal's algorithm \cite{Kruskal.1956}
\begin{algorithm}
\SetAlgoRefName{Parallel-MST}
{\bf var} $\CGG$: array[$1$..$m$] of $\{0,1\}$ initially  $\forall j: \CGG[j] := 0$;\\
// Edges $\CGG$ are assumed to be in increasing order of weights\\
{\bf forall} edges $e_i=(v_j,v_k)$: do {\bf in parallel}\\
  \h {\bf if} there is no path from $v_j$ to $v_k$ with edges $1..j-1$ then\\
\h \h  $\CGG[j] := 1$;
 \caption{Finding the minimum spanning tree in a graph in NC \label{fig:alg-kruskal-llp}.}
\end{algorithm}
In this example, the predicate $B$ on a set of edges $G$ is defined to be
``$G$ forms the minimum spanning tree in the graph.'' Since all edge weights are unique,
$B$ is a solitary predicate. In the Boolean lattice of all edges,
an edge from $v_j$ to $v_k$ is forbidden if there is no path from $v_j$ to $v_k$ using 
edges with weight lower than the weight of the edge $(v_j, v_k)$. 
Such an edge is always included as part of the minimum spanning tree.
Furthermore, the predicate satisfies antimonotone advancement property.
If an edge $(v_j, v_k)$ is not forbidden, then it continues to stay not forbidden.
Finally, the poset corresponding to a Boolean lattice has $O(1)$ height, and
therefore by Theorem \ref{thm:antimono}, we have an \NC~algorithm.
Many \NC~algorithms are already known for the minimum
spanning tree problem. Our goal was to show how the notion of solitary predicates in a lattice with
$O(1)$ height and {\em antimonotone} advancement leads to an \NC~algorithm.

We now discuss another property of predicate that allows us to detect it in \NC: {\em efficient
rejection}. A solitary predicate that has the efficient rejection property
can be detected in \NC~even when the height of the poset is not $O(\log^k(n))$.
We show that the problem of finding the
least global state that satisfies the conjunctive predicate and the problem of 
reachability in a directed graph satisfy the efficient rejection property.

As a concrete example, consider {\em conjunctive predicates}. We are searching for the least global state that satisfies 
$$B_{conj} = l_1 \wedge l_2 \wedge \ldots l_n.$$
Since we are only interested in the least global state, we can view it as a solitary predicate
with $B_{conj}$ appropriately refined as
$$B_{fconj}(G) = B_{conj}(G) \wedge \forall H: B_{conj}(H) \Rightarrow (G \leq H).$$

We start with the notion of a rejection relation. A state $s$ rejects a state $t$ if whenever all local states less than or equal to $s$
are forbidden, then all states less than or equal to $t$ are also forbidden.
If we know the initial forbidden states and the rejection relation, we can compute the  states
on each process that are forbidden in the initial state or become forbidden when the processes are advanced 
on the forbidden states. This is done by computing the
reflexive transitive closure of the rejection relation. Then, the first local state that is not rejected on every process
gives us the least global state that satisfies the predicate.
Formally,
\begin{definition}[Rejection Relation]
  Given any distributive lattice $\CC$ derived from a poset $P$ with $n$ chains, $P_1, \ldots, P_n$ and a predicate $B$, for $s\in P_i$ and $t\in P_j$ we define,
  \begin{equation*}
      \begin{split}
          rejects(s,t)  \equiv \forall s' \preceq_i &s:  \forall G\in\CC: forbidden(s', G, B) \Rightarrow \\
          &\forall t' \preceq_j t: \forall H\in\CC:forbidden(t', H, B),
      \end{split}
  \end{equation*}
  where $\preceq_i$ is the reflexive order associated with the chain $P_i$.
\end{definition}

We say that a lattice-linear predicate satisfies {\em efficient rejection} if
\begin{definition}[Lattice-linear Predicate with Efficient Rejection]
A lattice linear Boolean predicate $B$ satisfies {\em {efficient rejection}} with respect to
a lattice $\CC$ if there exists a rejection relation such that $rejects(s,t)$ can be computed in \NC.
\end{definition}

We observe here that computing the rejection relation requires as input the entire poset. Hence, the algorithm
derived using {\em efficient rejection} is offline. For example, the algorithm for the 
conjunctive predicate derived using the rejection relation requires the entire poset and is offline.
This algorithm is in contrast to the algorithm derived using {\em efficient advancement} property
\cite{GargWald:WeakUnstable} which is online.

Algorithm \ref{fig:strong-LLP} detects a lattice-linear predicate with efficient rejection in \NC.
This algorithm is similar to the one proposed in \cite{DBLP:conf/icdcn/GargG19} where the algorithm
is proposed for conjunctive predicates and we refer the reader to \cite{DBLP:conf/icdcn/GargG19} for
details. We show that the algorithm is applicable to any
lattice-linear predicate with efficient rejection. Thus, it is also applicable to finding
the shortest path in a directed graph.

\begin{algorithm}
\SetAlgoRefName{LLPwithRejection}
Output: Consistent Global State as array $G[1 \ldots n]$ \\
{\bf var}\\
\h $G$: array$[1 \ldots n]$ of $1 \ldots m $ initially $1$;\\ 
\h $R$: array$[(1 \ldots n,1 \ldots m),(1 \ldots n,1 \ldots m)]$ of $0 \ldots 1$\\
\h\h\h initially $R[(i,j),(i',j')] = 1$ iff $(i,j)$ rejects $(i',j')$;\\
\h$valid$: array$[1 \ldots n][1\ldots m]$ of $ 0 \ldots 1$\\
\h\h\h initially $\forall i,j: valid[i][j] :=1$;\\
\h$RT:  array[(1 \ldots n,1 \ldots m),(1 \ldots n,1 \ldots m)]$ of $0 \ldots 1$;\\
\BlankLine
$RT := TransitiveClosure(R)$; \\
{\bf for all} $( i \in 1 \ldots n, i' \in 1 \ldots n, j' \in 1 \ldots m)$ {\bf in parallel}  do\\
\h {if} $forbidden(i,1)  \wedge (RT[(i,1), (i',j')] = 1)$ {then} \\
\h \h $valid[i'][j'] := 0$; \\
 {\bf for all} $(i \in 1 \ldots n)$ {\bf in parallel} do\\
\h  if $\forall j \in 1 \ldots m: (valid[i][j] = 0)$ then return false; \\
\h else $G[i] := \min~ \{ j~|~ valid[i][j]=1 \}$;
\caption{An \NC~algorithm to find the first consistent cut that satisfies an lattice-linear predicate with the rejection relation. \label{fig:strong-LLP}}
\end{algorithm}

We now show how lattice-linear predicates with efficient rejection can be applied to graph reachability.
Suppose that we are given a fixed vertex $v_0$ in a directed graph $(V,E)$. Our goal is to find all
the vertices that are reachable from $v_0$. 
We can use simple BFS to find all reachable vertices. 
Let $G[i]$ be a binary variable such that $G[i]$ is $1$ if $v_i$ is reachable from $v_0$.
Then, the following predicate is true on a binary vector $G$ that corresponds to the reachable vertices.
\[ B_{traverse}  \equiv (G[0] = 1) \wedge (\forall (v_i, v_j) \in E: G[j] \geq G[i]) \]
The predicate $B_{traverse}$ is a lattice-linear predicate and can be employed to find all the reachable vertices.
However, this procedure takes time proportional to the diameter of the graph.

By computing the reflexive transitive closure of the binary graph, we can find all reachable vertices
in polylogarithmic time. Let $A$ be the graph in the matrix form. We compute $\CGG$, the reflexive transitive closure of the matrix $A$, 
such that $\CGG[i,j]$ equals $1$ iff the vertex $v_j$ is reachable from the vertex $v_i$.
We define $B_{closure}$ as
$$\forall i,j: \CGG[i,j] \geq \max (A[i,j], \max  \{\CGG[i,k] \wedge \CGG[k,j] ~|~ k \in [0..n-1] \}).$$
For this problem, our poset has $n^2$ processes and each process has just one event.
The process $P_{(i,j)}$ is forbidden if $G[i,j]$ equals $0$ and $A[i,j]$ equals $1$ or for some $k$, $G[i,k]$ and $G[k,j]$ are both $1$.

We now claim that $B_{closure}$ is a lattice-linear predicate with efficient rejection. 
If for any $(i,j)$, $A[i,j]$ is $1$, then the state $G[k,i]$ rejects $G[k,j]$ (i.e.,
if $i$ is reachable from $k$, then $j$ is also reachable from $k$). Thus, whenever all states equal to or prior to $G[k,i]=0$ are forbidden, then 
so are the states equal to or prior to $(G[k,j]=0)$.
Thus, we have that
the predicate $B_{closure} \equiv \forall i,j: \CGG[i,j] \geq \max (A[i,j], \max  \{\CGG[i,k] \wedge \CGG[k,j] ~|~ k \in [0..n-1] \}) $ is a lattice-linear predicate with efficient rejection.

\section{Conclusions and Open Problems}

We have defined a class of predicates called Equilevel predicates on finite distributive lattices.
We have also identified subclasses of equilevel predicates that can be detected efficiently
in parallel. 
There are many problems that are open in this area.
Are there other subclasses of equilevel predicates or solitary 
predicates that admit efficient detection?
What other problems can be modeled as equilevel predicate detection and do they introduce any specific algorithmic challenges?
Are there techniques that can be used to find approximate solutions of the equilevel predicate detection problem?

On the other end of the spectrum, what properties preclude efficient parallel detection of equilevel predicates?
In this work, we made general statements about the \NP-hardness of detecting equilevel predicates.
However, similar to how we identified properties enabling parallel detection of these predicates, it would be insightful to identify properties that certify \P-hardness (see \cite{greenlaw1995limits}) of the detection of corresponding subclasses of equilevel predicates which, in turn, would rule out the possibility of an efficient parallel solution unless \NC=\P.

\section*{Acknowledgements}
We thank David Alves, and Rohan Garg for discussions on this topic. We would also like to thank anonymous reviewers of this paper for useful comments. This work was supported in parts by the National Science Foundation Grant CNS-1812349 and the Cullen Trust Endowed Professorship.



\bibliography{fmaster}

\begin{thebibliography}{10}

\bibitem{davey}
B.~A. Davey and H.~A. Priestley.
\newblock {\em Introduction to Lattices and Order}.
\newblock Cambridge University Press, Cambridge, UK, 1990.

\bibitem{DBLP:conf/spaa/Garg20}
Vijay~K. Garg.
\newblock Predicate detection to solve combinatorial optimization problems.
\newblock In Christian Scheideler and Michael Spear, editors, {\em {SPAA} '20: 32nd {ACM} Symposium on Parallelism in Algorithms and Architectures, Virtual Event, USA, July 15-17, 2020}, pages 235--245. {ACM}, 2020.

\bibitem{gale1962college}
David Gale and Lloyd~S Shapley.
\newblock College admissions and the stability of marriage.
\newblock {\em The American Mathematical Monthly}, 69(1):9--15, 1962.

\bibitem{gusfield1989stable}
Dan Gusfield and Robert~W Irving.
\newblock {\em The stable marriage problem: structure and algorithms}.
\newblock MIT press, 1989.

\bibitem{Dijkstra1959}
E.~W. Dijkstra.
\newblock A note on two problems in connexion with graphs.
\newblock {\em Numerische Mathematik}, 1(1):269--271, Dec 1959.

\bibitem{GargWald:WeakUnstable}
V.~K. Garg and B.~Waldecker.
\newblock Detection of weak unstable predicates in distributed programs.
\newblock {\em IEEE Trans. on Parallel and Distributed Systems}, 5(3):299--307, March 1994.

\bibitem{DBLP:conf/sss/Garg21}
Vijay~K. Garg.
\newblock A lattice linear predicate parallel algorithm for the housing market problem.
\newblock In Colette Johnen, Elad~Michael Schiller, and Stefan Schmid, editors, {\em Stabilization, Safety, and Security of Distributed Systems - 23rd International Symposium, {SSS} 2021, Virtual Event, November 17-20, 2021, Proceedings}, volume 13046 of {\em Lecture Notes in Computer Science}, pages 108--122. Springer, 2021.

\bibitem{DBLP:conf/sss/GuptaK21}
Arya~T. Gupta and Sandeep~S. Kulkarni.
\newblock Extending lattice linearity for self-stabilizing algorithms.
\newblock In Colette Johnen, Elad~Michael Schiller, and Stefan Schmid, editors, {\em Stabilization, Safety, and Security of Distributed Systems - 23rd International Symposium, {SSS} 2021, Virtual Event, November 17-20, 2021, Proceedings}, volume 13046 of {\em Lecture Notes in Computer Science}, pages 365--379. Springer, 2021.

\bibitem{Garg:ICDCN22}
Vijay~K. Garg.
\newblock A lattice linear predicate parallel algorithm for the dynamic programming problems.
\newblock In {\em Proc. of the Int'l Conf. on Distributed Computing and Networking (ICDCN)}, Delhi, India, 2022. Springer-Verlag.

\bibitem{AlvGar22}
David~R. Alves and Vijay~K. Garg.
\newblock Parallel minimum spanning tree algorithms via lattice linear predicate detection.
\newblock In {\em Proc. Parallel and Distributed Combinatorics and Optimization (PDCO), June 2022}, Lyon, France, 2022.

\bibitem{DBLP:conf/icdcn/Garg23}
Vijay~K. Garg.
\newblock Keynote talk: Lattice linear predicate algorithms for the constrained stable marriage problem with ties.
\newblock In {\em 24th International Conference on Distributed Computing and Networking, {ICDCN} 2023, Kharagpur, India, January 4-7, 2023}, pages 2--11. {ACM}, 2023.

\bibitem{DBLP:journals/corr/abs-2302-07207}
Arya~T. Gupta and Sandeep~S. Kulkarni.
\newblock Multiplication and modulo are lattice linear.
\newblock {\em CoRR}, abs/2302.07207, 2023.

\bibitem{gupta2023lattice}
Arya~Tanmay Gupta and Sandeep~S. Kulkarni.
\newblock Lattice linearity in assembling myopic robots on an infinite triangular grid.
\newblock {\em CoRR}, abs/2307.13080, 2023.

\bibitem{CoopMarz:ConsDetGP}
Robert Cooper and Keith Marzullo.
\newblock Consistent detection of global predicates.
\newblock {\em ACM SIGPLAN Notices}, 26(12):167--174, 1991.

\bibitem{GarWal:wpdd}
V.~K. Garg and B.~Waldecker.
\newblock Detection of unstable predicates.
\newblock In {\em Proc. of the Workshop on Parallel and Distributed Debugging}, Santa Cruz, CA, May 1991.

\bibitem{DBLP:journals/dc/ChaseG98}
Craig~M. Chase and Vijay~K. Garg.
\newblock Detection of global predicates: Techniques and their limitations.
\newblock {\em Distributed Comput.}, 11(4):191--201, 1998.

\bibitem{MG:dcs01slice}
V.~K. Garg and N.~Mittal.
\newblock On slicing a distributed computation.
\newblock In {\em 21st Intnatl. Conf. on Distributed Computing Systems ({ICDCS}' 01)}, pages 322--329, Washington - Brussels - Tokyo, April 2001. IEEE.

\bibitem{CDF95}
B.~Charron-Bost, C.~Delporte-Gallet, and H.~Fauconnier.
\newblock Local and temporal predicates in distributed systems.
\newblock {\em ACM Trans. on Programming Languages and Systems}, 17(1):157--179, January 1995.

\bibitem{Birk1}
G.~Birkhoff.
\newblock {\em Lattice Theory}.
\newblock Providence, R.I., 1940.
\newblock first edition.

\bibitem{CLR}
T.~H. Cormen, C.~E. Leiserson, R.~L. Rivest, and C.~Stein.
\newblock {\em Introduction to Algorithms}.
\newblock The {MIT} Press and {McGraw-Hill}, 2001.
\newblock second edition.

\bibitem{fenner2019deterministic}
Stephen Fenner, Rohit Gurjar, and Thomas Thierauf.
\newblock A deterministic parallel algorithm for bipartite perfect matching.
\newblock {\em Communications of the ACM}, 62(3):109--115, 2019.

\bibitem{Nesetril.2001}
Jaroslav Nešetřil, Eva Milková, and Helena Nešetřilová.
\newblock {Otakar Borůvka on minimum spanning tree problem Translation of both the 1926 papers, comments, history}.
\newblock {\em Discrete Mathematics}, 233(1-3):3--36, 2001.

\bibitem{DBLP:conf/stoc/ValiantV85}
Leslie~G. Valiant and Vijay~V. Vazirani.
\newblock {NP} is as easy as detecting unique solutions.
\newblock In Robert Sedgewick, editor, {\em Proceedings of the 17th Annual {ACM} Symposium on Theory of Computing, May 6-8, 1985, Providence, Rhode Island, {USA}}, pages 458--463. {ACM}, 1985.

\bibitem{DBLP:journals/ipl/KashyapG05}
Sujatha Kashyap and Vijay~K. Garg.
\newblock Intractability results in predicate detection.
\newblock {\em Inf. Process. Lett.}, 94(6):277--282, 2005.

\bibitem{shapley1974cores}
Lloyd Shapley and Herbert Scarf.
\newblock On cores and indivisibility.
\newblock {\em Journal of mathematical economics}, 1(1):23--37, 1974.

\bibitem{Kruskal.1956}
Joseph~B. Kruskal.
\newblock {On the shortest spanning subtree of a graph and the traveling salesman problem}.
\newblock {\em Proceedings of the American Mathematical Society}, 7(1):48--50, 1956.

\bibitem{DBLP:conf/icdcn/GargG19}
Vijay~K. Garg and Rohan Garg.
\newblock Parallel algorithms for predicate detection.
\newblock In R.~C. Hansdah, Dilip Krishnaswamy, and Nitin~H. Vaidya, editors, {\em Proceedings of the 20th International Conference on Distributed Computing and Networking, {ICDCN} 2019, Bangalore, India, January 04-07, 2019}, pages 51--60. {ACM}, 2019.

\bibitem{greenlaw1995limits}
Raymond Greenlaw, H~James Hoover, and Walter~L Ruzzo.
\newblock {\em Limits to parallel computation: P-completeness theory}.
\newblock Oxford University Press, USA, 1995.

\end{thebibliography}
\end{document}